%% file: if.tex
\documentclass[10pt]{amsart}
\usepackage{amsmath}
\usepackage{amsfonts}
\usepackage{amssymb}
\usepackage[dvips]{graphicx}
\usepackage{subfigure}
\input{nhgmacr.tex}

\newcommand{\sett}[1]{\ensuremath{\left \{ #1 \right \}}}
\newcommand{\abs}[1]{\ensuremath{\left| #1 \right| }}

\newcommand{\pwo}{\ensuremath{PW_\Omega}}

\newtheorem{theo}{Theorem}
\newtheorem{prop}{Proposition}
\newtheorem{obs}{Observation}
\newtheorem{assum}{Assumption}
\newtheorem{claim}{Claim}
\newtheorem{rem}{Remark}
\newtheorem{example}{Example}

\title[The integrate and fire sampler]{Approximate reconstruction of
bandlimited functions for the integrate and fire sampler}

\author[H.G.~Feichtinger]{Hans G. Feichtinger}
\address{Faculty of Mathematics, University of Vienna\\ Nordbergstrasse 15\\
Vienna, Austria}
\email[Hans G. Feichtinger]{hans.feichtinger@univie.ac.at}

\author[J.~C.~Principe]{Jos\'e C. Pr\'incipe}
\address{Department of Electrical and Computer Engineering \\ University of
Florida
\\ Gainesville, FL 32611, USA}
\email[J.~C.~Pr\'incipe]{principe@cnel.ufl.edu}

\author[J.L.~Romero]{Jos\'e Luis Romero}
\address{Departamento de
Matem\'atica \\ Facultad de Ciencias Exactas y Naturales\\ Universidad
de Buenos Aires\\ Ciudad Universitaria, Pabell\'on I\\ 1428 Capital
Federal\\ Argentina\\ and CONICET, Argentina}
\email[Jos\'e Luis Romero]{jlromero@dm.uba.ar}

\author[A.~S.~Alvarado]{Alexander Singh Alvarado}
\address{Department of Electrical and Computer Engineering \\ University of
Florida
\\ Gainesville, FL 32611, USA}
\email[Alexander~Singh~Alvarado]{asingh@cnel.ufl.edu}

\author[G.~Velasco]{Gino Angelo Velasco}
\address{Faculty of Mathematics, University of Vienna\\ Nordbergstrasse 15\\
Vienna, Austria\\ and Institute of Mathematics, University of the Philippines Diliman\\ Quezon City, Philippines}
\email[Gino Angelo Velasco]{gino.velasco@univie.ac.at}

\begin{document}
\maketitle
\begin{abstract}
In this paper we study the reconstruction of a bandlimited signal from samples generated by the integrate and fire model. This sampler allows us to
trade complexity in the reconstruction algorithms for simple hardware
implementations, and is specially convenient in situations where the
sampling device is limited in terms of power, area and bandwidth.

Although perfect reconstruction for this sampler is impossible, we give
a general approximate reconstruction procedure and bound the corresponding
error. We also show the performance of the proposed algorithm through
numerical simulations.

\end{abstract}

\section{Introduction}
The integrate and fire (IF) model is well known in computational
neuroscience as a simplified model of a neuron \cite{Gerstner:2002p4013, Rieke:1997} and is typically used to study the dynamics
of large populations. The model consists of a leaky integrator followed by a
comparator. The leak corresponds to a gradual loss of the value of the
integral.

More recently, the IF model has also been considered as a sampler
\cite{Chen:0p2586,Wei:2005p2180,Lazar:2005p3764,Lazar:2003p1620}, where the
sampler output is tuned to the variation of the integral of the signal. This
feature can be exploited when sampling neural recordings, for which relevant
information is localized in small intervals where the signal has a high
amplitude \cite{singh09}.

The block diagram of the sampler is presented in Figure \ref{fig:IF_diagram}.
At every instant $s$, the continuous input $x(t)$ is integrated against an
averaging function $u_{k,s}(t)$ and the result is compared to a positive
and negative threshold. When either of these is reached, a pulse is created
at time $t_{k}=s$ representing the threshold value (positive or negative),
the value of the integrator is then reset and the process repeats. 
\begin{figure}[ht]
	\begin{center}
\includegraphics[angle = -90, width=0.4\textwidth]{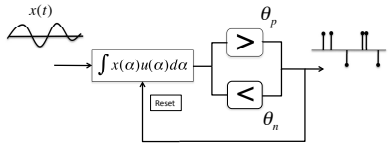}
\caption{Block diagram for the BIF model.}
\label{fig:IF_diagram}
	\end{center}
\end{figure}
The output is a nonuniformly spaced pulse train, where each of the pulses is
either 1 or -1. The averaging function $u_{k,s}(t)$ is defined by
$e^{\alpha(t-s)}\mathcal{X}_{[t_{k}, s]}$, where $\mathcal{X}_{I}$ is the
characteristic function of $I$ and $\alpha > 0$ is a constant that models the
leakage of the integrator due to practical implementations. The precise
firing condition determining the pulses is:
\begin{equation}
\label{uk}
\pm \theta = \int_{t_{k}}^{t_{k+1}} f(t)
e^{\frac{-(t_{k+1}-t)}{\alpha}} dt =: \inner{f}{u_k}.
\end{equation}
The simplicity of the sampler translates into an efficient hardware
implementation which saves both power and area when compared to conventional
analog-to-digital converters (ADC) \cite{Chen:0p2586}. These constraints are
severe, in the case of wireless brain machine interfaces
\cite{Sanchez:2008}, for which the entire system has to be embedded inside
the subject. Hence, the IF sampler allows us to move the complexity of the
design into the reconstruction algorithm while providing a simple front end
at the sampling stage.

The problem of reconstructing a signal from the IF output should be
distinguished from the study of the dynamics of a population of neurons,
when some stochastic assumption is made on the firing parameters
\cite{lazar09}. In this article we study the deterministic reconstruction of
a bandlimited signal from the integrate and fire output. Part of the
challenge of this stems from the fact that the sampling map that
associates a function to its samples is non-linear. Indeed,
we see from Equation \eqref{uk} that the magnitude of the samples is
always $\theta$. Moreover, exact reconstruction for the IF sampler is
impossible since the output of the sampler does not completely determine the
signal (see Example \ref{no_fire} below.) 

In this article, we will show that it is however possible to approximately
reconstruct a bandlimited signal in $L^\infty$ norm with an error comparable
to the threshold $\theta$. Moreover, we give a concrete reconstruction
procedure which is of course non-linear but, nevertheless, easy to implement.
Since in many situations the IF sampler is so much more convenient
to implement than conventional analog-to-digital converters, the loss of
accuracy in the reconstruction is a very reasonable trade-off
\cite{Chen:0p2586}, specially if the final analysis of the reconstructed
data tolerates some small error \cite{singh09}.

The methods considered so far \cite{Lazar:2003p1620} reconstruct the signal
$f$ from the system of equations $\inner{f}{u_k} = \pm \theta$ (cf. Equation
\eqref{uk}), thus treating the reconstruction as a (linear) average sampling
problem (see \cite{fegr94, al02, suzh02-4, sc03}). These approaches impose
density restrictions on the set of sampling functions $\sett{u_k}_k$
(cf. Equation \eqref{uk}.)
Since these sampling functions depend on the signal, the density constraints
on them are somehow unnatural.

The key for the reconstruction method that we develop lies in the
observation that the information derived from the IF output is much richer
than the mere system of equations $\inner{f}{u_k} = \pm \theta$. It also
contains the information
that no proper subinterval $[t_k, t']$ of $[t_k,t_{k+1}]$ satisfies Equation
\eqref{uk}. We will exploit this extra information to give an approximate
reconstruction procedure for a general bandlimited function. Since the
sampling process starts at a certain instant $t_0$, an
additional assumption on the size of $f$ before $t_0$ is required in order to
fully reconstruct $f$. Roughly speaking, the assumption means that the
sampling scheme would not have produced any pulse before $t_0$. 

In Section \ref{sec_problem} we formally describe the output of the IF
sampling scheme. This output depends on an initial time $t_0$ when the
process is started and two parameters: the threshold $\theta$ and the
constant $\alpha>0$ modeling the leakage on the sampler. In Section
\ref{sec_output} we show that the IF output is always a finite sequence.
Section \ref{sec_rec} gives the approximate reconstruction procedure and 
Section \ref{sec_num} presents some numerical experiments.

\section{The integrate and fire sampling problem}
\label{sec_problem}
We now define precisely the integrate and fire sampling scheme. Throughout
the article we will assume the following.
\begin{assum}
\label{assum_1}
A bandlimited function $f \in \pwo$ and numbers $t_0 \in \Rst$, $\alpha,
\theta>0$ are given.
\end{assum}
Here, $\pwo$ is the Paley-Wiener space
\[
\pwo := \set{f \in
L^2(\Rst)}{\supp(\hat{f}) \subseteq [-\Omega,\Omega]},
\]
of (complex-valued) bandlimited functions
and $\hat{f}(w) := \int_\Rst f(x) e^{-2\pi i w x} dx$ is the Fourier
transform of $f$. We call $t_0$ the \emph{initial time}, $\alpha$ the
\emph{firing
parameter}
and $\theta$ the \emph{threshold}. Using these parameters we formally define
the output of the sampler. We first define recursively a finite or countable
sequence $t_0 < \ldots < t_j \ldots$ called \emph{the time instants}. Suppose
that the instants $t_0 < \ldots < t_j$ have already been defined and consider
the function $F_j:[t_j,+\infty) \to \bC$ given by
\[
F_j(t) := \int_{t_j}^t f(x) e^{\alpha(x-t)} dx.                       
\]
Observe that $F_j$ is continuous and $F_j(t_j)=0$. If
$\abs{F_j(t)} < \theta$, for all $t \geq t_j$, then the process stops.
If $\abs{F_j(t)} \geq \theta$, for some $t \geq t_j$, by the continuity of
$F_j$, we can define $t_{j+1}$
as the minimum number satisfying the equation
\begin{equation}
\label{fire_tj}
\abs{\int_{t_j}^{t_{j+1}} f(x) e^{\alpha(x-{t_{j+1}})} dx} = \theta.
\end{equation}
Clearly, in this case $t_{j+1} > t_j$.

We have defined a finite or countable sequence of points $t_0 <
\ldots < t_j \ldots$. We will prove in Proposition \ref{samples_interval}
that this sequence is in fact finite. Let us assume the time instants $\sett{t_0, \ldots, t_n}$ and define the \emph{samples} $\sett{q_1, \ldots, q_n}$ by,
\begin{equation}
\label{qj}
q_j := \int_{t_{j-1}}^{t_j} f(x) e^{\alpha(x-t_j)} dx,
\qquad (1 \leq j \leq n).
\end{equation}
Observe that, by the definition of the time intervals, $\abs{q_j} = \theta$.

The output of the sampler is formally given by the time instants $\sett{t_0,
\ldots, t_n}$ and the numbers
$\sett{q_1, \ldots, q_n}$. We say that this output has been produced by the
\emph{integrate and fire}.
The succeeding results apply generally to complex-valued functions, but in the case of the application that motivated this sampling scheme, the
signal is taken to be real-valued, and
the output of the sampler is encoded as a train of
impulses, where only the sign of the samples $q_j$ is stored.

\section{Some remarks on the IF output}
First we note that bandlimited functions are not completely determined by
the output of the IF sampler.
\begin{example}
\label{no_fire}
There are non-zero bandlimited signals that will never
produce an output from the sampler. Take for
instance $f_{\theta}(x)=\frac{\theta\,\sin^2(\pi\,x)}{2\pi^2\,x^2}$.
Since
$\hat{f_{\theta}}(\omega)=\frac{\theta}{2}\max\{1-|\omega|,0\}$, $f_{\theta}$
is bandlimited. We have for any $t_0\in\mathbb{R}$,
$\left|\int_{t_0}^{t}f_{\theta}(x)\,e^{\alpha(x-t)}dx\right|\leq\int_{t_0}^{
t}\left|\frac{\theta\,\sin^2(\pi\,x)}{2\pi^2\,x^2}\right|\,dx\leq\frac{\theta
}{2}\int_{\mathbb{R}}\frac{\sin^2(\pi\,x)}{\pi^2\,x^2}\,dx=\frac{\theta}{2}
<\theta,$
for all $t\geq t_0.$
\end{example}

\label{sec_output}
We now prove that the set of time instants produced by the IF sampler is
indeed finite and give some
bounds on its distribution. To this end we introduce some auxiliary functions
that will be used throughout the remainder of the article.

Consider the function $g: \Rst \to \Rst$ given by $g(x) = e^{-\alpha
x}\chi_{[0, \infty]}$ and define,
\begin{align}
\label{deff_v}
v(t) := (f*g)(t) := \int_{-\infty}^t f(x) e^{\alpha(x-t)} dx.
\end{align}
Since $g \in L^1(\Rst)$, $v \in \pwo$. In the Fourier domain, $v$ and $f$ are
related by
\begin{align}
\label{fv_fourier}
\hat{f}(w) = \left(2 \pi i w + \alpha\right) \hat{v}(w).
\end{align}
In the time domain, this can be expressed as
\begin{align}
\label{fv_time}
f(t) = \frac{\partial v(t)}{\partial t} + \alpha v(t).
\end{align}

\begin{obs}
\label{v_is_c0}
The function $v$ is continuous and $v(t) \longrightarrow 0$, when $t
\longrightarrow \pm \infty$.
\end{obs}
\begin{proof}
We have already observed that $v \in \pwo$. Since $\hat{v} \in L^2$ and
$\supp(\hat{v}) \subseteq [-\Omega, \Omega]$, we have that $\hat{v} \in
L^1$ and the conclusion follows from the Riemann-Lebesgue Lemma.
\end{proof}
The following straightforward equation relates $v$ to the integrate and fire
process.
\begin{equation}
\label{v_if}
\int_s^t f(x) e^{\alpha(x-t)} dx = v(t) - e^{\alpha(s-t)} v(s),
\quad s \leq t.
\end{equation}

We can now prove that the output of the IF process is finite.
\begin{prop}
\label{samples_interval}
Under Assumption 1, the following holds.
\begin{itemize}
 \item[(a)] The set of time instants produced by the integrate and fire
scheme is a finite set
$\sett{t_0, \ldots, t_n}$.
 \item[(b)] The numbers of time instants $t_j$ in a given finite interval
$[a,b]$ is bounded by
\[
\frac{\norm{f}_2}{\theta} (b-a)^{1/2} +1.
\]
 \item[(c)] If $f$ is integrable, the total number of time instants is
bounded by
\[
\frac{\norm{f}_1}{\theta}+1.
\]
\end{itemize}
\end{prop}
\begin{proof}
We first prove (b) and (c). Let $[a,b]$ be an interval and let $\sett{t_j,
\ldots, t_{j+m-1}}$ be $m$ consecutive time instants contained in $[a,b]$. If
$m \leq 1$ the bound is trivial, so assume that $m \geq 2$. For each $0 \leq
k \leq m-2$, using Equation \eqref{fire_tj}
we have,
\begin{align*}
\theta &= \abs{\int_{t_{j+k}}^{t_{j+k+1}} f(x) e^{\alpha(x-t_{j+k+1})} dx}
\\
&\leq \int_{t_{j+k}}^{t_{j+k+1}} \abs{f(x)} dx.
\end{align*}
Summing over the $m-1$ intervals determined by the points $\sett{t_j, \ldots,
t_{j+m-1}}$ we have,
\begin{align}
\label{estimated_number_of_points}
(m-1) \theta &\leq \int_a^b \abs{f(x)} dx.
\end{align}
Letting $a=-\infty$ and $b=+\infty$ yields (b). For (a), H\"older's
inequality gives,
\[
(m-1) \theta \leq \norm{f}_2 (b-a)^{1/2}, 
\]
and the conclusion follows.

Now we prove (a). Assume on the contrary that the IF process goes on forever
producing an infinite set of instants $\sett{t_j: j \geq 0}$. Given $s>t_0$,
by part (b), only a finite number of instants $t_j$ belong to $[t_0, s]$.
Therefore $t_n \rightarrow +\infty$, as $n \rightarrow +\infty$. Using
Equations \eqref{v_if} and \eqref{fire_tj} it follows that,
\begin{align*}
\theta &= \abs{\int_{t_j}^{t_{j+1}} f(x) e^{\alpha(x-t_{j+1})} dx}
\\
&= \abs{v(t_{j+1}) - e^{\alpha(t_j-t_{j+1})} v(t_j)}
\leq \abs{v(t_{j+1})} + \abs{v(t_j)}.
\end{align*}
This contradicts Observation \ref{v_is_c0}.
\end{proof}

\section{The reconstruction}
\label{sec_rec}
We now address the problem of approximately reconstructing a bandlimited
function from the integrate and fire output. Since the samples are taken in
the half-line $[t_0,+\infty)$ we will make some assumption about the size of
$f$ before the initial instant. Roughly speaking, the integrate and fire process would not have produced any sample in the
interval $(-\infty,t_0]$.
\begin{assum}
\label{assum_2}
The function defined in Equation \eqref{deff_v} satisfies,
\[
\abs{v(t)} \leq \theta,
\mbox{ for all $t \leq t_0$.} 
\]
\end{assum}
Note that by Observation \ref{v_is_c0}, any $t_0 \ll 0$ satisfies this
assumption. To approximately reconstruct $f$ we will first approximately
reconstruct $v$ from the integrate and fire output and then derive
information about $f$ by means of Equation \eqref{fv_fourier}. We will use
the structure of the IF process to produce a number of approximate samples
for $v$.

First we argue that, from the output of the IF process, we have enough
information to approximate $v$ on the time instants $\sett{t_0, \ldots,
t_n}$.
Rewriting Equation \eqref{qj} in terms of $v$ 
(cf. Equation \eqref{v_if}) we have,
\begin{equation}
\label{rec_vtj}
v(t_{j+1}) =  e^{\alpha(t_j-t_{j+1})}v(t_j) + q_{j+1},
\quad (0 \leq j \leq n-1).
\end{equation}
Since the value $v(t_0)$ may not be exactly known we cannot determine from
this recurrence relation all the values $v(t_j)$. However, we can construct
an approximation to these values. Let $w_0 := 0$ and define recursively,
\begin{equation}
\label{rec_wj}
w_{j+1} =  e^{\alpha(t_j-t_{j+1})}w_j + q_{j+1},
\quad (0 \leq j \leq n-1).
\end{equation}
Observe that Assumption \ref{assum_2} implies that $\abs{w_0-v(t_0)} \leq
\theta$. Using this estimate as a starting point we can iterate on Equation
\eqref{rec_vtj} and \eqref{rec_wj} to get,
\begin{equation}
\label{estimate_w}
\abs{w_j - v(t_j)} \leq \theta,
\quad (0 \leq j \leq n).
\end{equation}
Consequently, using only the output of the IF sampling scheme, we have
constructed a set of values $\sett{w_0, \ldots, w_n}$ that approximates $v$
on
the instants $\sett{t_0, \ldots, t_n}$. The second step is to approximate $v$
on
an arbitrary point of $\Rst$.

To this end observe that, according to the definition of $t_j$
as the minimum number satisfying Equation \eqref{fire_tj}, we have that,
\begin{equation*}
\abs{\int_{t_j}^t f(x) e^{\alpha(x-t)} dx} \leq \theta,
\quad
\mbox{for all $t \in [t_j,t_{j+1}]$.}
\end{equation*}
Rewriting this inequality in terms of $v$ (cf. Equation \eqref{v_if}) gives,
\begin{equation}
\label{band_almost}
\biggl | v(t) - e^{\alpha(t_j-t)}v(t_j) \biggr | \leq \theta,
\quad
\mbox{for all $t \in [t_j,t_{j+1}]$.}
\end{equation}
Combining this last inequality with \eqref{estimate_w} yields,
\begin{equation}
\label{band}
\biggl | v(t) - e^{\alpha(t_j-t)} w_j \biggr | \leq 2 \theta,
\quad
\mbox{for all $t \in [t_j, t_{j+1}]$.}
\end{equation}
We now show that this inequality allows us to approximate $v$ anywhere on the
line.
\begin{claim}
\label{claim_approx_v}
Given an arbitrary time instant $t \in \Rst$, choose $x \in \Rdst$ in the
following way:
\begin{itemize}
 \item[(a)] if $t < t_0$, let $x:=0$,
 \item[(b)] if $t$ belongs to some (unique) interval $[t_j,t_{j+1})$, let
$x := e^{\alpha(t_j-t)} w_j$,
 \item[(c)] if $t \geq t_n$, let $x:= e^{\alpha(t_n-t)} w_n$.
\end{itemize}
Then, $\abs{v(t)-x} \leq 2 \theta$.
\end{claim}
\begin{rem}
Observe that the procedure to obtain $x$ from $t$ depends only on the output
of the IF process. 
\end{rem}

\begin{proof}
For case (a), the conclusion follows from Assumption \ref{assum_2}. For case
(b), the conclusion follows from Inequality \eqref{band}. For case (c), the
fact that the fire condition is never satisfied after $t_n$ gives,
\begin{equation}
\biggl | v(t) - e^{\alpha(t_n-t)}v(t_n) \biggr | \leq \theta.
\end{equation}
Combining this estimate with Inequality \eqref{estimate_w}, the conclusion
follows.
\end{proof}

We will now choose a window function.
\begin{assum}
\label{assum_3}
A Schwartz class function $\psi$ such that
\begin{itemize}
 \item $\hat{\psi} \equiv 1$ on $[-\Omega,\Omega]$, and,
 \item $\hat{\psi}$ is compactly supported,
\end{itemize}
has been chosen.
\end{assum}
Since $v \in \pwo$, the classic oversampling trick for bandlimited functions
(see for example \cite{fe92-3} or \cite{fegr94}) implies that there exists a
number $0<\beta<(2 \Omega)^{-1}$ such that
\begin{equation}
\label{expansion_v}
v = \sum_{k\in\Zst} v(\beta k) \psi(\cdot-\beta k).
\end{equation}

Using the procedure described in Claim \ref{claim_approx_v}, we produce a set
$\sett{s_k}_{k \in \Zst}$
such that
\begin{equation}
\label{v_approx}
\abs{v(\beta k) - s_k} \leq \ 2 \theta,
\mbox{ for all $k\in\Zst$}.
\end{equation}

Let $\varphi$ be the function defined by,
\begin{equation}
\label{deff_phi}
\hat{\varphi}(w) = \left(2 \pi i w + \alpha\right) \hat{\psi}(w).
\end{equation}
It follows that $\varphi$ is also a Schwartz function. Moreover, using
Equation
\eqref{fv_fourier} we have that,
\begin{equation}
\label{expansion_f}
f = \sum_{k\in\Zst} v(\beta k) \varphi(\cdot-\beta k).
\end{equation}
Observe that, since $v \in \pwo$, the sequence $\sett{v(\beta k)}_k \in
\ell^2$ and the series in Equation \eqref{expansion_f} converges in $L^2$ and
uniformly - in fact, it converges in the Wiener amalgam norm $W(C_0, L^2)$, 
see for example \cite{fe92-3}, \cite{fegr94} and \cite{algr01}.)

Now we can define the approximation of $f$ constructed from the IF samples.
Let,
\begin{equation}
\label{approx_f}
\tilde{f} := \sum_{k\in\Zst} s_k \varphi(\cdot-\beta k).
\end{equation}
Since, by Inequality \eqref{v_approx}, the sequence $\sett{s_k}_k$ is bounded
and $\varphi$ is a Schwartz function, it follows that Equation
\eqref{approx_f} defines a bounded function and that the convergence is
uniform (see \cite{fe92-3} or \cite{algr01}.) 

The reconstruction algorithm consists then of calculating the approximated
samples $\sett{s_k}_k$ following Claim 1 and then convolving them with the
kernel $\varphi$, that can be pre-calculated.

We now give a precise error bound for the reconstruction.
\begin{theo}
\label{main_th}
Under Assumptions 1, 2 and 3, the function defined by Equation
\eqref{approx_f}
satisfies,
\[
\norm{f-\tilde{f}}_\infty \leq C \theta,
\]
for some constant C that only depends on $\Omega$ and the window function
chosen in Assumption \ref{assum_3}.
\end{theo}
\begin{proof}
According to Equations \eqref{expansion_f}, \eqref{approx_f} and Inequality \eqref{v_approx},
\begin{align*}
\norm{f-\tilde{f}}_\infty &\leq
\supess \sum_{k\in\Zst} \abs{v(\beta k)-s_k} \abs{\varphi(\cdot-\beta k)}
\\
&\leq 2\theta \supess \sum_{k\in\Zst} \abs{\varphi(\cdot-\beta k)}.
\end{align*}
It suffices to define $C:= 2 \sup \sum_{k\in\Zst} \abs{\varphi(\cdot-\beta
k)}$.
Since $\varphi$ is a Schwartz function, $C<+\infty$ (see for example \cite{fe83}).
\end{proof}
\begin{rem}
We currently do not know what choice of the window function $\psi$ minimizes
the constant in the theorem. A more detailed study of the choice of the
window function should not only consider the size of that constant but also
the rate of convergence of the series in Equation \eqref{approx_f}.
\end{rem}
\section{Numerical experiments}
\label{sec_num}
We study the behavior of the reconstruction algorithm under variations in the threshold and the oversampling period for a specific choice of reconstruction kernel $\psi$.
The test signal $f$ is of finite length and real valued, produced as a
linear combination of five `sinc' kernels ($\sin(\pi x)/(\pi x)$)
at a 1Hz frequency, with random locations and weights. The amplitude of
the input has been normalized to 1. Although the theory covers infinite dimensional spaces, our simulations are limited by the practical implementations of the sampler and algorithms. The effects of truncation and quantization are not considered here. 

This signal is encoded by the IF sampler with $\alpha =1$ and recovered using the procedure described in Section \ref{sec_rec}.  
The reconstruction kernel $\psi$ is a raised cosine, defined by,
\begin{equation}
\psi(t) =
	\sinc(t/T_s) \cos(\pi\gamma
t/T_s)\biggl ( 1-\frac{4\gamma^2t^2}{{T_s}^2}\biggr )^{-1},
\end{equation}
where $\gamma = 0.5$ and $T_s=0.25$ are determined by the maximum input frequency $\Omega$ and the desired oversampling period $\beta$ (cf. Equation
\eqref{expansion_v}.) Figure \ref{fig:raisedCos} shows the raised cosine
$\psi$ in the time and frequency domain. 
Observe that the spectrum of $\psi$ is constant for frequencies less than the
input bandwidth and then decays smoothly towards zero.
The corresponding kernel $\varphi$ (cf. Equation \eqref{deff_phi}) is shown in Figure~\ref{fig:raisedCosfilt}.
\begin{figure}[ht]
\centering
\subfigure[Kernel $\psi$.]{
\includegraphics[width=0.45\textwidth]{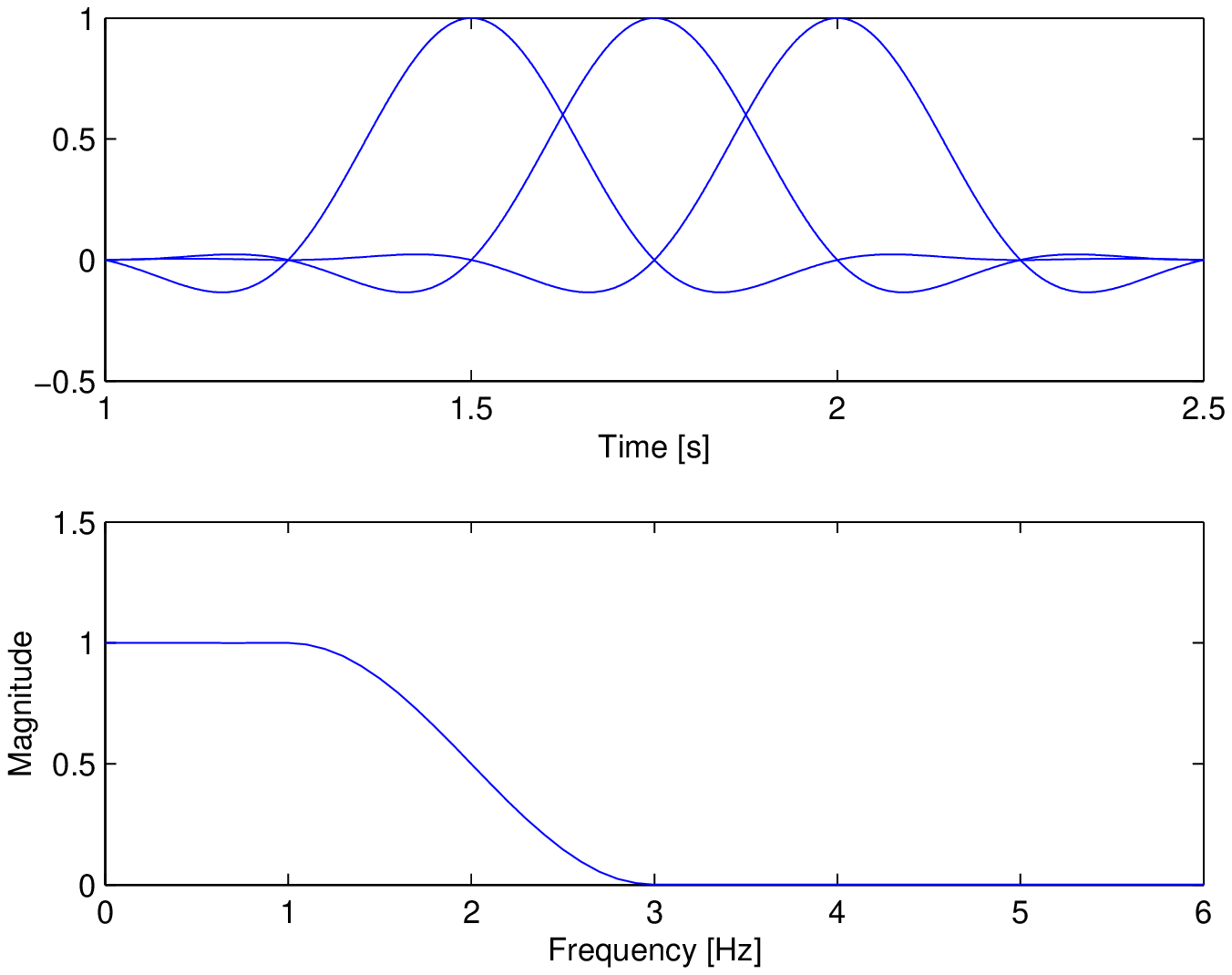}
\label{fig:raisedCos}
}
\subfigure[Kernel $\varphi$.]{
\includegraphics[width=0.45 \textwidth]{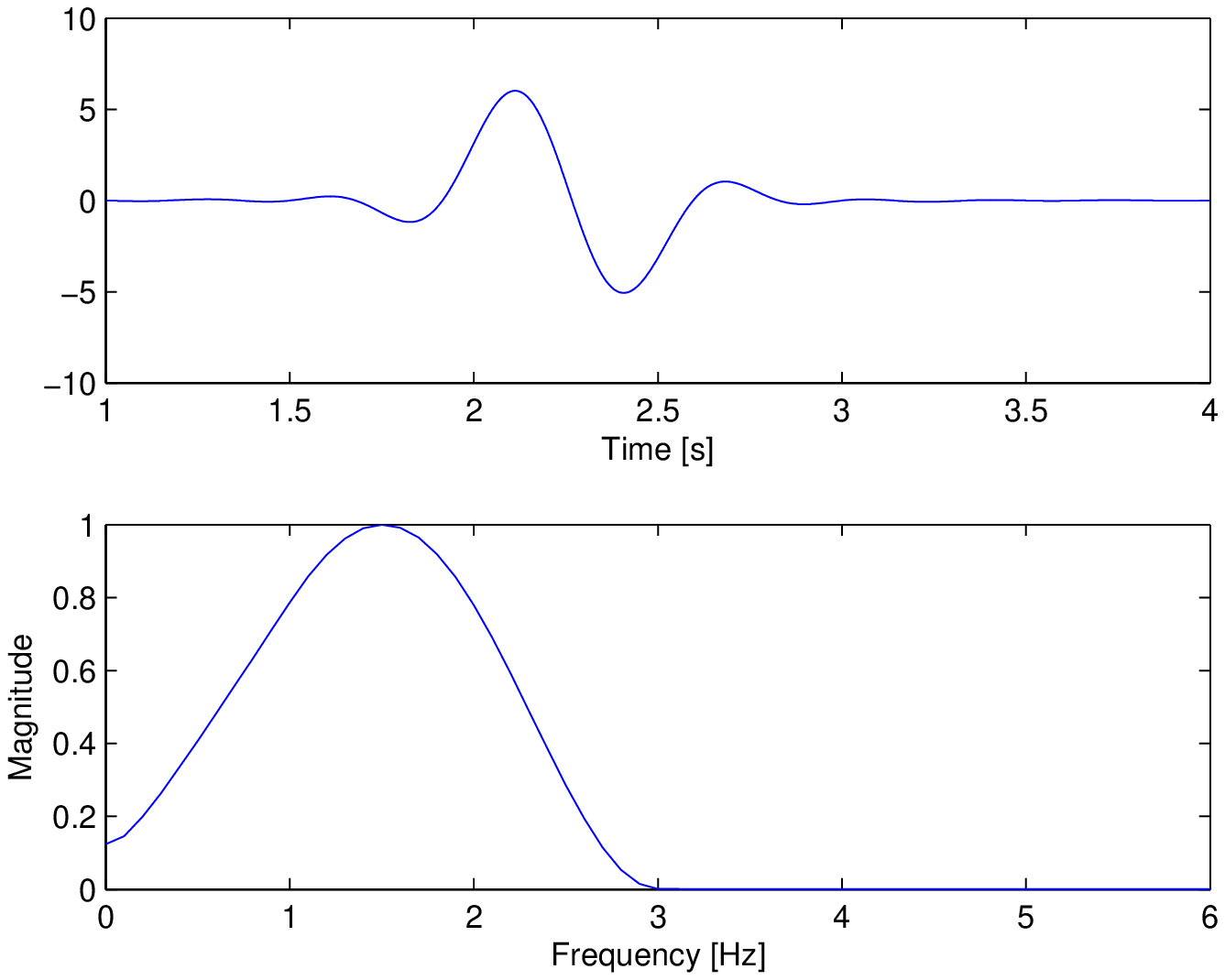} \label{fig:raisedCosfilt}
}
\caption{Reconstruction kernels.}
\end{figure}
Using $\varphi$ we recover $\tilde{f}$ (cf. Equation \eqref{approx_f}.), an approximation of $f$ as shown in Figure \ref{fig:timeRec_decLat}.
As expected the error decreases in regions with high density of samples. This behavior is evident from Figure \ref{fig:approxV}, the dense regions imply that the uniform samples will most likely coincide with the estimated values of $v(t)$ at the sample locations. On the other hand, for samples that are far apart the approximation follows a exponential decay from its original value which is not the natural trend in the signal.
Figure \ref{fig:approxV} shows $v(t)$ (solid line), and the approximated samples of $v$ on the lattice $\beta \Zst$, constructed using the procedure described in Claim \ref{claim_approx_v} called $\sett{s_k}_k$ and the envelope $v(t) \pm \theta$ where these samples are known to lie (dashed line.) 

Currently the reconstruction algorithm uses the approximated samples of $v(t)$ at the pulse locations to define the piecewise exponential bound and estimate the reconstruction coefficients on the uniform lattice. Based on the numerical experiments the algorithm can be improved by including the estimated value of $v(t)$ at the pulse locations although it implies reconstruction on a nonuniform grid.
\begin{figure}[hb]
	\begin{center}
\includegraphics[width=0.7\textwidth]{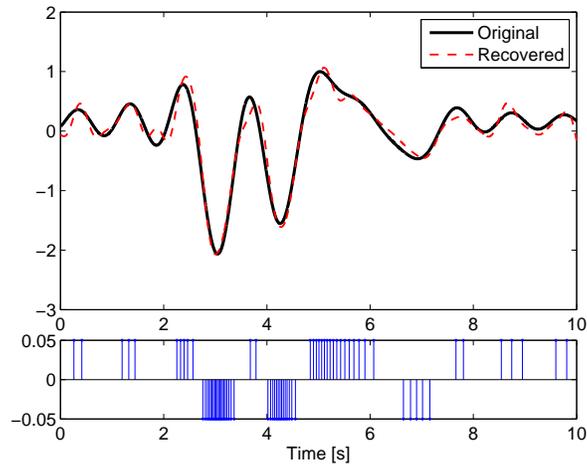}
\caption{Reconstruction of $f(t)$ from the impulse train.}
\label{fig:timeRec_decLat}
	\end{center}
\end{figure}
\begin{figure}[hb]
	\begin{center}
\includegraphics[width=0.7\textwidth]{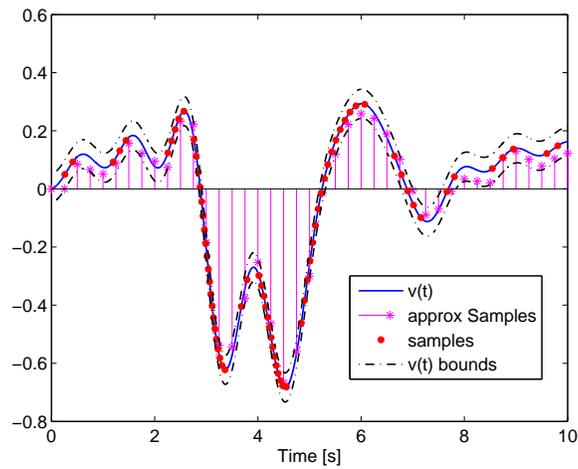}
\caption{Reconstruction of $v(t)$ with $\beta = 1/4$, $\theta = 0.05$}
\label{fig:approxV}
	\end{center}
\end{figure}
For both cases similar error bounds can be defined as in Theorem \ref{main_th}. The variation of the error in relation to the threshold (pulse rate) is shown in Figure \ref{fig:error}. 
\begin{figure}[tb]
	\begin{center}
\includegraphics[width=0.7\textwidth]{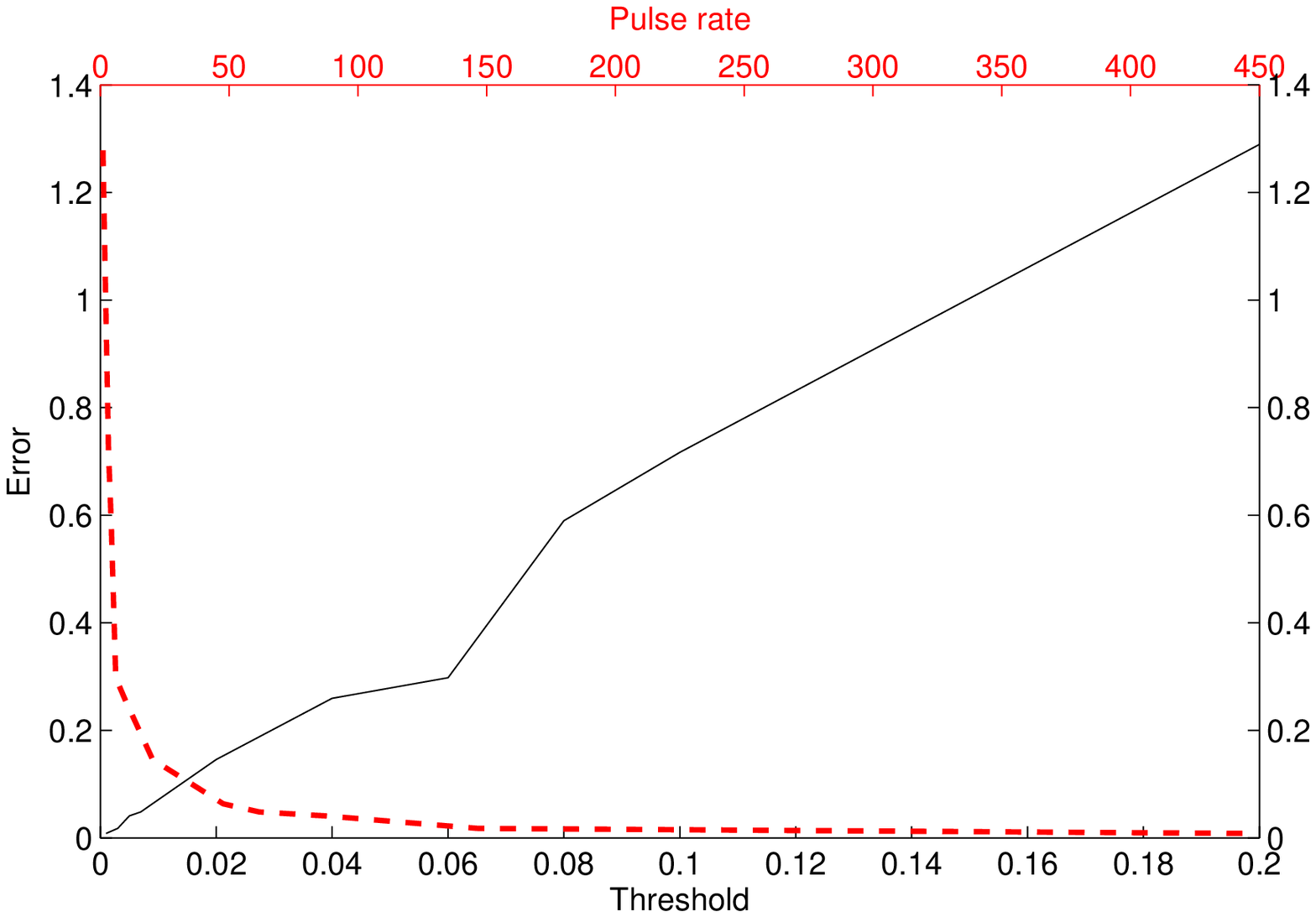}
\caption{Variation of the error $||f - \widehat{f}||_{\infty}$ in relation to
the threshold and pulse rate (dotted line) with $\beta = 0.25$.}
\label{fig:error}
	\end{center}
\end{figure}
The error depends on the choice of generator and the oversampling period $\beta$, as seen in Figure~\ref{fig:errorvsBeta}. The relationship between the kernels and the optimal oversampling period is still not evident.
\begin{figure}[ht]
	\begin{center}
\includegraphics[width=0.7\textwidth]{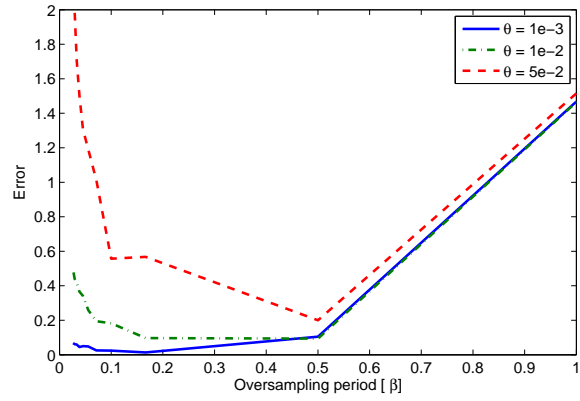}
\caption{Variation of the error in relation to the oversampling period for different thresholds. The error is defined as $||f - \widehat{f}||_{\infty}$.}
\label{fig:errorvsBeta}
	\end{center}
\end{figure}
\section{Acknowledgements}
The second and fourth authors were supported by NINDS (Grant Number: NS053561). The third author was partially supported by the following grants:
PICT06-00177, CONICET PIP 112-200801-00398 and UBACyT X149. The third and fourth authors' visit to the Numerical Harmonic Analysis Group (NuHAG) of the University of Vienna was funded by the European Marie Curie Excellence Grant EUCETIFA FP6-517154.
\clearpage
\bibliographystyle{abbrv}
\bibliography{if}
\end{document}

%% file: nhgmacr.tex
\newcommand{\bC}{\field{C}}        

\def\Cessum{{\lower 1.5pt \hbox{\large C}}{-}\!\!\!\sum}
\def\CHI{\hbox{\raise .5ex \hbox{$\chi$}}}

\newcommand{\field}[1]{\mathbb{#1}}   

\newcommand{\GLsp}{{\boldsymbol G\kern-0.1em\boldsymbol L}}

\newcommand{\inner}[2]{\langle #1,#2\rangle}   

\def\Msp{{\boldsymbol M}}

\def\mv1{\Msp_v^1}   

\newcommand{\norm}[1]{\lVert#1\rVert}   

\def\Rdst{{\Rst^d}}   

\def\Rst{{\mathbb R}}

\def\sabs{{{\raise 0.5pt \hbox{\footnotesize $|$}}}}

\newcommand{\set}[2]{\big\{ \, #1 \, \big| \, #2 \, \big\}}   

\def\shah{{\makebox[2.7ex][s]{$\sqcup$\hspace{-0.5em}\hfill$\sqcup$}}}   
\def\shah1{{\makebox[2.3ex][s]{$\sqcup$\hspace{-0.15em}\hfill $\sqcup$}}}   

\def\sinc{\operatorname{sinc}}   

\def\slb{{{\raise 0.5pt \hbox{\footnotesize $[$}}}}
\def\slcb{{{\raise 0.5pt \hbox{\footnotesize $\{$}}}}
\def\slp{{{\raise 0.5pt \hbox{\footnotesize $($}}}}

\newcommand{\snorm}{{{\raise 0.5pt \hbox{\footnotesize $\|$}}}}

\def\srb{{{\raise 0.5pt \hbox{\footnotesize $]$}}}}
\def\srcb{{{\raise 0.5pt \hbox{\footnotesize $\}$}}}}
\def\srp{{{\raise 0.5pt \hbox{\footnotesize $)$}}}}

\newcommand{\supess}{\mathop{\operatorname{supess}}}   
\def\supp{\operatorname{supp}}   

\def\Zst{{\mathbb Z}}

%% file: if.bbl
\begin{thebibliography}{10}

\bibitem{al02}
A.~{A}ldroubi.
\newblock {N}on-uniform weighted average sampling and reconstruction in
  shift-invariant and wavelet spaces.
\newblock {\em {A}ppl. {C}omput. {H}armon. {A}nal.}, 13(2):151--161, 2002.

\bibitem{algr01}
A.~{A}ldroubi and K.~{G}r{\"o}chenig.
\newblock {N}onuniform sampling and reconstruction in shift-invariant spaces.
\newblock {\em {S}{I}{A}{M} {R}ev.}, 43(4):585--620, 2001.

\bibitem{singh09}
A.~Alvarado, J.~Principe, and J.~Harris.
\newblock Stimulus reconstruction from the biphasic integrate-and-fire sampler.
\newblock In {\em Neural Engineering, 2009. NER '09. 4th International
  IEEE/EMBS Conference on}, pages 415--418, April 29 2009-May 2 2009.

\bibitem{Chen:0p2586}
D.~Chen, Y.~Li, D.~Xu, J.~Harris, and J.~Principe.
\newblock Asynchronous biphasic pulse signal coding and its cmos realization.
\newblock {\em Circuits and Systems, 2006. ISCAS 2006. Proceedings. 2006 IEEE
  International Symposium on}, pages 4 pp.--2296, 0-0 2006.

\bibitem{fe83}
H.~G. {F}eichtinger.
\newblock {B}anach convolution algebras of {W}iener type.
\newblock In B.~{S}z. {N}agy and J.~{S}zabados, editors, {\em {P}roc. {C}onf.
  on {F}unctions, {S}eries, {O}perators, {B}udapest 1980}, volume~35 of {\em
  {C}olloq. {M}ath. {S}oc. {J}anos {B}olyai}, pages 509--524, {A}msterdam,
  1983. {N}orth-{H}olland.

\bibitem{fe92-3}
H.~G. {F}eichtinger.
\newblock {W}iener amalgams over {E}uclidean spaces and some of their
  applications.
\newblock In K.~{J}arosz, editor, {\em {F}unction {S}paces, {P}roc {C}onf,
  {E}dwardsville/{I}{L} ({U}{S}{A}) 1990}, volume 136 of {\em {L}ect. {N}otes
  {P}ure {A}ppl. {M}ath.}, pages 123--137, {N}ew {Y}ork, 1992. {M}arcel
  {D}ekker.

\bibitem{fegr94}
H.~G. {F}eichtinger and K.~{G}r{\"o}chenig.
\newblock {T}heory and practice of irregular sampling.
\newblock In J.~{B}enedetto and M.~{F}razier, editors, {\em {W}avelets:
  {M}athematics and {A}pplications}, {S}tudies in {A}dvanced {M}athematics,
  pages 305--363, {B}oca {R}aton, {F}{L}, 1994. {C}{R}{C} {P}ress.

\bibitem{Gerstner:2002p4013}
W.~Gerstner and W.~Kistler.
\newblock Spiking neuron models.
\newblock {\em Cambridge University Press}, 2002.

\bibitem{lazar09}
A.~A. Lazar and E.~A. Pnevmatikakis.
\newblock Reconstruction of sensory stimuli encoded with integrate-and-fire
  neurons with random thresholds.
\newblock {\em EURASIP Journal on Advances in Signal Processing}, 2009, July
  2009.

\bibitem{Lazar:2005p3764}
A.~A. Lazar, E.~K. Simonyi, and L.~T. Toth.
\newblock A {T}oeplitz formulation of a real-time algorithm for time decoding
  machines.
\newblock In {\em Proceedings of the Conference on Telecommunication Systems,
  Modeling and Analysis}, November 2005.

\bibitem{Lazar:2003p1620}
A.~A. Lazar and L.~T. Toth.
\newblock Time encoding and perfect recovery of bandlimited signals.
\newblock In {\em IEEE International Conference on Acoustics, Speech and Signal
  Processing}, volume~6, pages VI709--712, April 2003.

\bibitem{Rieke:1997}
F.~Rieke, D.~Warland, de~Ruyter, and W.~Bialek.
\newblock {\em Spikes: Exploring the Neural Code}.
\newblock The MIT Press, 1997.

\bibitem{Sanchez:2008}
J.~C. Sanchez, J.~C. Principe, T.~Nishida, R.~Bashirullah, J.~G. Harris, and
  J.~A.~B. Fortes.
\newblock Technology and signal processing for brain-machine interfaces.
\newblock {\em Signal Processing Magazine, IEEE}, 25(1):29--40, 2008.

\bibitem{sc03}
H.~{S}chwab.
\newblock {\em {R}econstruction from {A}verages}.
\newblock PhD thesis, {D}ept. {M}athematics, {U}niv. {V}ienna, 2003.

\bibitem{suzh02-4}
W.~{S}un and X.~{Z}hou.
\newblock {R}econstruction of band-limited functions from local averages.
\newblock {\em {C}onstr. {A}pprox.}, 18(2):205--222, 2002.

\bibitem{Wei:2005p2180}
D.~Wei.
\newblock {\em Time based analog to digital converters}.
\newblock PhD thesis, University of Florida, 2005.

\end{thebibliography}
